\documentclass[reqno]{amsart}
\usepackage{amssymb}
\usepackage{enumerate}
\usepackage{hyperref}
\usepackage{doi,url}
\usepackage{amsfonts}
\usepackage{color}
\usepackage{mathrsfs}

\definecolor{purple}{rgb}{.5,0,1}
\definecolor{orange}{rgb}{1,.5,0}
\definecolor{pink}{rgb}{1,0,.5}

\newcommand{\nn}{\notag}

\numberwithin{equation}{section}
\newtheorem{theorem}{Theorem}[section]
\newtheorem{proposition}[theorem]{Proposition}
\newtheorem{lemma}[theorem]{Lemma}

\newtheorem{definition}[theorem]{Definition}
\newtheorem{remark}[theorem]{Remark}

\newenvironment{acknowledgement}{\emph{Acknowledgement.}}

\DeclareMathOperator{\tr}{tr}

\newcommand\C{\mathbb C}
\newcommand\Z{\mathbb Z}

\newcommand\M{\mathcal{M}}

\newcommand\cA{\mathcal{A}}
\newcommand\cE{\mathcal{E}}

\newcommand\cN{\mathcal{N}}

\newcommand\cG{\mathcal{G}}
\newcommand\cV{\mathcal{V}}

\newcommand\cH{\mathcal{H}}

\newcommand\cK{\mathcal{K}}

\newcommand\cR{\mathcal{R}}

\newcommand\beq{\begin{equation}}
	\newcommand\eeq{\end{equation}}

\newcommand{\pa}[1]{\left( #1 \right)}

\newcommand{\tfd}{\pa{1- \tfrac{1}{\Delta}}}

\begin{document}
	
	\title[Entanglement Entropy for the XXZ-Spin Strip]{Entanglement Entropy Bounds for Droplet States of the XXZ Model on the Strip}
	
	\author{Christoph Fischbacher} 
	\address[C. Fischbacher]{Department of Mathematics, 
		Baylor University, Sid Richardson Bldg., 1410 S.\,4th Street, Waco, TX 76706, USA}
	\email{c\_fischbacher@baylor.edu}
	
	\author{Lee Fisher}
	\address[L. Fisher]{University of California, Irvine;
		Department of Mathematics;
		Irvine, CA 92697-3875,  USA}
	\email{lfisher2@uci.edu}
	
	\begin{abstract} The scaling behavior of the entanglement entropy of droplet states in Heisenberg spin-$1/2$ XXZ model defined on a strip of width $M$ under the presence of a non-negative background magnetic field is investigated. Without any assumptions on $V$, a logarithmically corrected area law is shown. Assuming that the values of $V$ are i.i.d.~ random variables, an area law in expectation is obtained.  
	\end{abstract}
	
	\maketitle
	\section{Introduction}
	
	In this paper, we will show bounds on the scaling behavior of the entanglement entropy (EE) for droplet states of the Heisenberg Spin-$1/2$ XXZ model defined on a strip of arbitrary width $M$. Without any further assumptions on the background magnetic field, we will show that the EE is bounded from above by a term that scales like the logarithm of the system size. On the other hand, under the presence of a random background magnetic field, we show an area law in expectation, thus indicating localization.
	This is a generalization of the work by Beaud and Warzel \cite{BW18}, who showed the same kind of result (logarithmic upper bound without and area law with randomness) for the one-dimensional chain (corresponding to $M=1$). Logarithmic upper bounds for the scaling of the EE in the one-dimensional chain have also been shown for higher-energy states \cite{ARFS} and also for higher local spins \cite{FO}. For the same models (spin $1/2$ and also higher local spins) localization results have also been obtained in \cite{BW17}, \cite{EKS2}, \cite{EKS}, and \cite{FisherDiss}. While by now, there are numerous results on EE bounds (to name a few of the most recent ones: \cite{MPS20, MS20, Pfeif, PfeifSpitz}), 
	there still seems to be a lack of such type of results for interacting many-particle systems in higher dimensions (as, for example, is pointed out in \cite{Sto20}). In this sense, we view our results as a first step towards this direction.
	
	We will proceed as follows:
	
	In Section \ref{sec:2}, we give the necessary background information to formulate our main results. We try to be as brief as possible, while providing  references to more detailed presentations.
	
	In Section \ref{sec:3}, we state our main results. The first result, Theorem \ref{thm:combinatorial}, is combinatorial and the key tool that enables us to show the subsequent bounds on the scaling behavior of the EE. These are stated in Theorems \ref{thm:logbound} and \ref{thm:arealaw}. 
	
	After this, in Section \ref{sec:4}, we present previous results that will be needed for the subsequent proofs of the main theorems: a deterministic bound on spectral projections (Proposition \ref{prop:spectralbound}), a preliminary deterministic estimate on the EE (Proposition \ref{prop:EEbound}), and an estimate on the expected value of the spectral projections (Lemma \ref{lemma:largedev}).
	
	We then introduce the notion of level sets in Section \ref{sec:5} -- these allow 
	
	Lastly, the proofs of the main theorems \ref{thm:combinatorial}, \ref{thm:logbound}, and \ref{thm:arealaw} are given in Section \ref{sec:6}.
	\section{The model and relevant concepts} \label{sec:2}
	\subsection{The XXZ Hamiltonian on general graphs}
	We consider the spin-$1/2$ XXZ model on general graphs and will prove estimates for the scaling behavior of the entanglement entropy of droplet states. Let $\cG=(\cV,\cE)$  be a countable, connected, and undirected graph of bounded maximal degree. We interpret the vertices in $\cV$ as the locations of the spins while the edge set $\cE$ describes where interactions between spins are present. Let $V:\cV\rightarrow [0,\infty)$ be an arbitrary non-negative function, which we refer to as the \emph{background potential}. Then the XXZ Hamiltonian $H_\cG(V)$ acting on the Hilbert space $\cH_\cG=\bigotimes_{u\in\cV}\C^2$ is formally given by
	\begin{equation} \label{eq:XXZ}
		H_\cG(V)=\sum_{\{u,v\}\in\cE}\left(-\frac{1}{\Delta}(S^1_uS^1_v+S_u^2S_v^2)+\frac{1}{4}-S^3_uS^3_v\right)+\sum_{u\in\cV}V(u)\left(\frac{1}{2}-S_u^3\right)\:,
	\end{equation}  
	where $\Delta>1$ and $S^1, S^2$, and $S^3$ are the spin-$1/2$ matrices given by
	\begin{equation}
		S^1=\begin{pmatrix}0 & 1/2\\1/2 & 0\end{pmatrix},\quad S^2=\begin{pmatrix}0 & -i/2\\i/2 & 0\end{pmatrix},\mbox{ and } S^3=\begin{pmatrix}1/2 & 0\\0 & -1/2\end{pmatrix}\:.
	\end{equation}
	For an arbitrary $A\in\C^{2\times 2}$, we use the usual convention that $A_u$ acts as $A$ on the component of the tensor product corresponding to $u\in\cV$ and as identity on all other factors. Note that if $\cG$ is an infinite graph, then it requires more effort to rigorously define $\cH_\cG$ and $H_\cG$. For more details on this, we refer to \cite{ARFS}, where an explicit construction and also a description of $H_\cG$ as a direct sum of discrete many-particle Schr\"odinger operators is given.    
	
	An important feature of $H_\cG(V)$ is the conservation of total magnetization or particle number. Defining the total particle number operator $\cN_\cG:=\sum_{u\in\cV}\left(\frac{1}{2}-S_u^3\right)$, it is easily verified that $[H_\cG,\cN_\cG]=0$. It is also easy to see that $\sigma(\cN_\cG)\subset\{0,1,2,\dots\}$. One then interprets the eigenspace of $\cN_\cG$ corresponding to an $N\in\sigma(\cN_\cG)$ as the ``$N$- down-spin" or ``$N$-particle subspace" -- denoted by $\cH_\cG^N$. Moreover, we introduce the notation $H_\cG^N(V):=H_\cG(V)\upharpoonright_{\cH_\cG^N}$, but, whenever convenient, our notation will suppress the dependence on $\cG$ and $V$ and we will just write $H, H^N,$ etc.

	\subsection{The XXZ Model on the strip} For this paper, we consider the XXZ model on the strip of length $2\ell$ and width $M$, i.e. $\cV=\{1,2,\dots,2\ell\}\times\{1,2,\dots,M\}$. 
	For $x, y\in \cV$, let $d(x,y):=|x-y|_1$. Then, the edge set $\cE$ of this graph is given by $\mathcal{E}=\{\{u,v\}\subset \cV: d(y,v)=1\}$. For later convenience, we also introduce the infinite strip of width $M$ given by $\mathcal{V}_\mathbb{Z}:=\mathbb{Z}\times\{1,2,\dots,M\}$.
	
	We will prove entanglement entropy bounds with respect to the bipartition $\cV=\Lambda_\ell\cup\Lambda_\ell^c$, where 
	\begin{equation}
		\Lambda_\ell:=\{1,2,\dots,\ell\}\times\{1,2,\dots,M\}\:.
	\end{equation}
	
	Moreover, for any $\delta>0$, introduce the \emph{droplet interval} (cf.~\cite{FS})
	\begin{equation}
		I_\delta=\left[\tfd M,\tfd (M+1-\delta)\right]\:.
	\end{equation} 
	For the results presented in this paper, we restrict our considerations to particle numbers $N$ that are of the form $N=kM$, where $k\in\{M,M+1,\dots\}$ and introduce the corresponding set of particle numbers given by
		\begin{equation} \mathfrak{N}:=\{kM:k\in\{M,M+1,M+2,\dots\}\}\:.
	\end{equation}
	 Our main results will be statements about droplet states that belong to $\cH_\cG^N(V)$, where $N\in\mathfrak{N}$. To this end, let $\mathcal{Q}_\delta(V):=\mbox{ran}\left(1_{I_\delta}(H_{\cG}(V))\right)$ denote the spectral subspace of $H_\cG(V)$ corresponding to $I_\delta$. Moreover, let $Q_\delta(V)$ denote the orthogonal projection onto $\mathcal{Q}_{\delta}(V)$.
	We refer to the elements of $\mathcal{Q}_\delta$ as \emph{droplet states}. We also define the restriction to the $N$-particle subspace: $Q_\delta^N=Q_\delta\upharpoonright_{\mathcal{H}^N}$ and $\mathcal{Q}_\delta^N:=\mbox{ran}(Q_\delta^N)$.
%
\begin{remark} Some of these assumptions are made for simplicity of presentation. With our methods, we can also obtain similar results for more general quasi-one-dimensional graphs such as $\mathbb{Z}\times\{1,2,\dots,M\}^d$ or more generally: $\mathbb{Z}\times \cK$, where $\cK$ is finite. In addition, we can also get bounds on the entanglement entropy of states corresponding to a slightly larger energy interval and also for other particle numbers (not just rectangular ones). The main difference is instead of the distance to the nearest rectangular configuration in Prop.~\ref{prop:spectralbound} has to be replaced by the distance to slightly more complicated configurations -- this just complicates the combinatorial arguments given in Theorem \ref{thm:combinatorial} and Lemma \ref{lemma:Ax-q1d} below.
	Lastly, it is also possible to consider bipartitions, where $|\Lambda_\ell|=q|\mathcal{V}|$ for any $q\in(0,1)$. Again, for simplicity we chose to only present the case $q=1/2$.
	
\end{remark}
	\subsection{Entanglement entropy} Let $\cH_{\Lambda_\ell}$ be the Hilbert space corresponding to the subregion $\Lambda_\ell\subset\cV$ and decompose $\cH_\cG=\cH_{\Lambda_\ell}\otimes\cH_{{\Lambda_\ell}^c}$. For any normalized state $\psi\in\cH_\cG$ with associated density matrix $\psi\langle\psi,\cdot\rangle$, let the reduced state $\rho_{\psi;\ell}:\cH_{\Lambda_\ell^c}\rightarrow\cH_{\Lambda_\ell^c}$ be given by the partial trace of $\psi\langle\psi,\cdot\rangle$ over $\cH_{\Lambda_\ell}$ (for more details cf.\ \cite{ARFS} or \cite{FO}). The Entanglement Entropy $\mathscr{E}(\psi;\ell)$ of $\psi$ with respect to the bipartition $\cV=\Lambda_\ell\cup\Lambda_\ell^c$ is then given by the von Neumann entropy of $\rho_{\psi;\ell}$, i.e.\
	\begin{equation}
		\mathscr{E}(\psi;\ell)=-\tr(\rho_{\psi;\ell}\log\rho_{\psi;\ell})\:.
	\end{equation}
	Moreover, for any $\alpha\in(0,1)$, we let $\mathscr{E}_\alpha(\psi;\ell):=\frac{1}{1-\alpha}\log\tr(\rho_{\psi;\ell}^\alpha)$ denote the $\alpha$-R\'enyi entropy of $\rho_{\psi;\ell}$. It is well-known that $\mathscr{E}\leq\mathscr{E}_\alpha$ for any $\alpha\in(0,1)$.
	
	\subsection{Configurations}
	
	To describe configurations of down-spins, we introduce the following notations: 
	Let $X=\{x_1,x_2,\dots,x_N\}$ and $Y=\{y_1,y_2,\dots,y_N\}$ both be subsets -- or ``$N$-particle configurations" -- of $\cV$, each with exactly $N$ elements (in the following, also referred to as ``particles"). The distance $d_N(\cdot,\cdot)$ between any such two $N$-particle configurations is given by (cf. \cite{FS})
	\begin{equation}
		d_N(X,Y)=\min_{\pi\in S_N} \sum_{i=1}^N d(x_i,y_{\pi(i)})\:,
	\end{equation}
	where $S_N$ is the group of permutations of $N$ elements.
	If $\mathcal{A, B}$ are sets of $N$-particle configurations, then $d_N(\mathcal{A},\mathcal{B}):=\min_{X\in\mathcal{A}, Y\in\mathcal{B}}d_N(X,Y)$.
	
	For any configuration $X\subset \cV$, we introduce
	\begin{equation}
		\chi_{\{X\}}:=\left(\prod_{u\in X} \left(\frac{1}{2}-S^3_u \right)\right)\left(\prod_{u\in \cV\setminus X} \left(\frac{1}{2}+S^3_u \right)\right)\:,
	\end{equation}
	which is a rank-one orthogonal projection onto the linear space spanned by the vector corresponding to the configuration with down-spins exactly at the vertices in $X$ and up-spins everywhere else (cf.\ also \cite{MMNP}).
	Moreover, for any set of such configurations $\mathcal{A}$, we introduce
	\begin{equation}
		\chi_\mathcal{A}:=\sum_{X\in\mathcal{A}}\chi_{\{X\}}\:.
	\end{equation}
	If $N$ is an integer multiple of $M$, i.e.~ $N=kM$, where $k\in\mathbb{N}$, let $R_i^N$ denote the following rectangular configuration of particles:
	\begin{equation}
		R_i^N=\{i,i+1,\dots, i+k-1\}\times\{1,2,\dots,M\}\:.
	\end{equation} 
For any $z\in\mathbb{Z}$, we will also employ the following notation to describe translations of rectangles: $R^N_i+z:=R_{i+z}^N$.
	Moreover, $\mathcal{R}^N$ denotes the set of \emph{all} such rectangular configurations on $\mathcal{V}$. Note that for any set of $N$-particle configurations $\mathcal{A}$, this implies
	\begin{equation}
		d_N(\cA,\cR^N)=\min_{z\in\mathbb{Z}}d_N(\cA, R+z)\:,
	\end{equation}
where $R\subset\mathcal{V}_\mathbb{Z}$ is an arbitrary rectangular configuration such that $|R|=N$. Lastly, given a non-negative background potential $V$, we to introduce the set of all rectangular configurations, where the potential is sufficiently small:
\begin{equation}
	\mathcal{R}_V^N=\left\{R\in\mathcal{R}_N: \sum_{x\in R}V(x)<1\right\}\:.
\end{equation}
	\section{Main results} \label{sec:3}
	
	We will now state our main theorems. The principal new contribution in this paper that enables us to show the subsequent theorems is the following combinatorial result:
	\begin{theorem} \label{thm:combinatorial}
		Let $R=R_i^N$ be any rectangular configuration with $N=|R|\in\mathfrak{N}$. Then for any $\mu>0$, we have
		\begin{equation}
			f(R,\mu):=\sum_{X\subset \mathcal{V}_\mathbb{Z}: |X|=N} e^{-\mu d_N(X,R)}\leq \left(1+\frac{2M}{\mu}\right)e^{\frac{4M}{\mu}}\:.
		\end{equation}
	\end{theorem}
\begin{remark}
	
\end{remark}
	The next theorem is a logarithmic upper bound for droplet states of $H^N$, where $N\in\mathfrak{N}$, which is independent of the background potential. 
	\begin{theorem} \label{thm:logbound}
		For any non-negative $V$ and any $\delta>0$ we have 
		\begin{equation} \label{eq:logbound}
			\limsup_{\ell\rightarrow\infty}\left(\frac{\sup_{N\in\mathfrak{N}}\sup\left\{\mathscr{E}(\psi;\ell):{\psi\in \mathcal{Q}^N_\delta, \|\psi\|=1}\right\}}{\log\ell} \right)\leq 1\:.
		\end{equation}

	\end{theorem}
	The last result treats the case when $V$ is random:
	\begin{theorem} \label{thm:arealaw}
		Let $\delta>0$ and let $\{\nu_x\}_{x\in\mathcal{V}_\mathbb{Z}}$ be a sequence of independent, identically distributed non-negative random variables such that $\mathbb{P}(\nu=0)\neq 1$. Let the random background potential be given by $V_\omega(x)=\nu_x$. Then, there exists a constant $K=K(\Delta,\delta)<\infty$ such that
		\begin{equation}
			\sup_\ell\mathbb{E}\left[{\sup_{N\in\mathfrak{N}}\sup\left\{\mathscr{E}(\psi;\ell):{\psi\in \mathcal{Q}^N_\delta, \|\psi\|=1}\right\}}\right]\leq K\:.
		\end{equation}
	\end{theorem}
	
	\section{Previous results needed for the proofs} \label{sec:4}
	
	To prove our main results, we need several propositions that have been shown previously. Some of these have been stated for more general graphs, some of these were shown for the one-dimensional chain (corresponding to $M=1$). Since the modifications that have to be made to treat the case $M>1$ are straightforward, we will only provide references to previously given proofs for Propositions \ref{prop:spectralbound} and \ref{prop:EEbound}. A version of Lemma \ref{lemma:largedev} can be found in \cite[Lemma 1.2]{BW17}, however our result is more general and follows from a simpler argument.
	\begin{proposition} \label{prop:spectralbound} Let $N\in\mathfrak{N}$, $\delta>0$, and $V$ an arbitrary non-negative background potential. Then, for any set of $N$-particle configurations $\mathcal{A}$, we obtain the following bound
		\begin{equation}
			\|\chi_\mathcal{A}Q_\delta^N(V)\|\leq C \exp(-\mu d_N(\mathcal{A,R}^N_V))\leq C \exp(-\mu d_N(\mathcal{A,R}^N))\:,
		\end{equation}
		where the constants $C$ and $\mu$ are given by
		\begin{equation}
			C=\frac{3\sqrt{5}}{2}\frac{(2M+1)^{3/2}}{\min\{1,\delta^{3/2}\}}\quad\mbox{ and }\quad \mu=\frac{1}{2}\log\left(1+\frac{\delta\Delta}{4M+2}\right)
		\end{equation}
	and $\mathcal{R}_V$ is the set of rectangular configurations
	\end{proposition}
	\begin{proof}
		This follows from the general result in \cite[Lemma 4.1]{ARFS}, using that -- according to \cite[Lemma B.2]{FS} -- rectangular configurations minimize the edge surface on the strip. See also \cite[Remark 3.4]{ARFS} concerning adding arbitrary non-negative background potentials. 
	\end{proof} 
	In order to state our next result, we introduce the set of configurations $\mathcal{A}_{X,N}$, which for any $N\in\mathbb{N}$, $X\subset\Lambda_\ell^c$, where $|X|\leq N$, is given by
	\begin{equation}
		\mathcal{A}_{X,N}=\{X\cup Z: Z\subset\Lambda_\ell, |X|+|Z|=N\}\:.
	\end{equation}

	\begin{proposition} \label{prop:EEbound}
		Let $\delta>0$, $N\in\mathfrak{N}$, and $\alpha\in(0,1)$. Then, for any normalized $\psi\in\mathcal{Q}^N_\delta$, we get the following estimate:
		\begin{equation} \label{eq:alphatrace}
			\tr\left(\rho_{\psi;\ell}^\alpha\right)\leq 6+2\sum_{j=1}^{|\Lambda_\ell^c|-1}\sum_{X\subset\Lambda_\ell^c: |X|=j}\left\|\chi_{\mathcal{A}_{X,N}}\psi\right\|^{2\alpha}\:.
		\end{equation}
	\end{proposition}
	\begin{proof}
		This follows from an argument completely analogous to the one given in \cite[Section 5, up to and including Equation (5.14)]{ARFS}.
	\end{proof}
	\begin{lemma} \label{lemma:largedev}
	Let $\{\nu_x\}_{x\in\mathcal{V}_\mathbb{Z}}$ be a sequence of independent, identically distributed non-negative random variables such that $\mathbb{P}(\nu=0)\neq 1$. Let the random background potential be given by $V_\omega(x)=\nu_x$. Then, there exist constants $\tilde{C}=\tilde{C}(\delta,\Delta,\mathbb{P})$ and $\lambda=\lambda(\delta,\Delta,\mathbb{P})<1$ such that for any $X\in\Lambda_\ell^c$, where $|X|=j$, we have
	\begin{equation}
		\mathbb{E}\left(\|\chi_{\mathcal{A}_{X,N}}Q_\delta^N(V_\omega)\|\right)\leq \tilde{C} \lambda^j\:.
	\end{equation}
\end{lemma}

\begin{proof} Since $\mathbb{P}(\nu=0)\neq 1$, there exist $k\in\mathbb{N}$ and $p>0$ such that $\mathbb{P}(\nu\geq 1/k)=p$. Now, for each $x\in\mathcal{V}_\mathbb{Z}$, introduce the Bernoulli random variable $Y_x$, which is equal to $1$ if $\nu_x\geq 1/k$ (with probability $p$) and zero else (with probability $1-p$).
Then, for any $X\subset\Lambda_\ell^c$ with $|X|=j$, we have
\begin{equation} \label{eq:distest}
	d_N(\mathcal{A}_{X,N},\mathcal{R}^N_{V_\omega})\geq \sum_{x\in X} Y_x(\omega)-(k-1)\:.
\end{equation}	
	To see this, note that $\sum_{x\in X}Y_x$ counts the number of sites in $X$, where the background potential has a value of at least $1/k$. Thus, if $\sum_{x\in X}Y_x\geq k$, this means that there can be no rectangular configuration $R\in\mathcal{R}^N_V$ such that $X\subset R$. At the very least, $\sum_{x\in X} Y_x(\omega)-(k-1)$ particles would each have to move by one or more steps to reach a configuration in $\cR_{V_\omega}^N$, thus implying \eqref{eq:distest}.  
	Consequently, using Proposition \ref{prop:EEbound}, we get
	\begin{align}
		&\mathbb{E}\left(\|Q_{\delta}^N(V_\omega)\chi_{\mathcal{A}_{X,N}}\|\right) \leq C\cdot\mathbb{E}(\exp(-\mu d_N(\cA_{X,N}, \mathcal{R}_{V_\omega}^N)))\\
		\leq& Ce^{\mu(k-1)}\mathbb{E}\left[\exp\left(-\mu\sum_{x\in X}Y_x(\omega)\right)\right]=\tilde{C}\left(1-p+pe^{-\mu}\right)^{|X|}=\tilde{C}\lambda^j\:,
	\end{align} 
where $\lambda:=(1-p+pe^{-\mu})<1$ and $\tilde{C}:=Ce^{\mu(k-1)}$. 
\end{proof}

\section{Level Sets}                        \label{sec:5}                                                                                                                                                                                                                                                                                                                                                                                                                                                                                                                                                                 
For a rectangular configuration $R^{kM}_i=\{i,i+1,\dots,i+k-1\}\times \{1,2,\dots,M\}$, let its \emph{internal boundary} be given by 
\begin{equation}
\partial^{in}R^{kM}_i=\{i\}\times \{1,2,\dots,M\}\cup \{i+k-1\}\times\{1,2,\dots,M\}\:.
\end{equation}
\begin{definition} 
	\label{def:levels}
	Let $R\subset\mathcal{V}_\mathbb{Z}$ be a rectangular configuration. We define the \emph{level function} $L_R:\mathcal{V}_\mathbb{Z}\rightarrow \mathbb{Z}$ to be given by
	\beq
	L_R(x) = \begin{cases}
		d(x, \partial^{in}R) &\text{ if } x\not\in R\\
		-d(x, \partial^{in}R) &\text{ if } x\in R\\
	\end{cases}
	\eeq	
	
	The preimages $L_R^{-1}(n)$ will be called the levels of $R$, and the quantities $L_n(R) = |L_R^{-1}(n)|$ are the sizes of the levels of $R$. Note that there are only finitely many non-positive levels.   
\end{definition}

\begin{definition}
	\label{def:level-enum} We say that an enumeration $e_n$ of $\mathcal{V}_\mathbb{Z}$ is \emph{level respecting} if 
	\begin{itemize}
		\item $e_1 \in L_R^{-1}(\min\{n \in \Z: L_R^{-1}(n)\neq\emptyset\})$.
		\item $L_R(e_n) \geq L_R(e_m)$ if and only if $n \geq m$.
	\end{itemize}
\end{definition}

\begin{lemma}
	\label{lem:est-w-levels}
	Let $R$ be a rectangular configuration with $|R|=N$ and let $e_n$ be a level respecting enumeration of $\mathcal{V}_\mathbb{Z}$. We have the following inequality:
	\begin{align}
		&\sum_{\substack{ X \subset \mathcal{V}_\mathbb{Z}\\ |X| = N}} e^{-\mu d_N(X,R)} \leq \sum_{j = 0}^N \left(\sum_{\substack{ X \subset \mathcal{V}_\mathbb{Z} \setminus R \\ |X| = j }} e^{-\mu \sum_{x\in X} d(x, \partial^{in}R )}\right) \left(\sum_{\substack{ Y \subset R \\ |X| = j }} e^{-\mu \sum_{y\in Y} d(y, \partial^{in}R )}\right) \nn \\
		=& \sum_{j = 0}^N \left(\sum_{N < x_1 <\dots < x_j} e^{-\mu(L_R(e_{x_1})+\cdots L_R(e_{x_j}) )}\right) \left(\sum_{0 < y_1 < \dots < y_j \leq N} e^{-\mu (|L_R(e_{y_1})|+\cdots |L_R(e_{y_j})|)}\right) 
	\end{align}
\end{lemma}
\begin{proof}
	Decompose $X$  into $X_{in} \cup X_{out}$ corresponding to the elements of $X$ that lie inside or outside of $R$. If the configuration $X$ moves to $R$ along a shortest path, then every element/particle of $X_{out}$ must eventually pass through $\partial^{in}R$ -- thus contributing to the distance $d_N(X,R)$ between $X$ and $R$ by at least $\sum_{x\in X} d(x, \partial^{in}R)$. Now, for the particles in $X_{in}$, i.e. those that already lie inside the rectangle $R$, consider $Y = R \setminus X_{in}$ -- the complimentary configuration of ``holes". Eventually, each such hole has to move to the outside of $R$, thus contributing to the distance $d_N(X,R)$ by at least $\sum_{y\in Y} d(y, \partial^{in}R )$. Therefore, 
	\beq
	d(X,R) \geq \sum_{x\in X} d(x, \partial^{in}R) + \sum_{y\in Y} d(y, \partial^{in}R ).
	\eeq
	
	The first inequality follows. The second equality follows from the definition of the level function, and of level respecting enumerations. In particular notice that every configuration of $j$ particles is realized as an ordered subset of $j$ elements of the enumeration $e_n$. Moreover, since $e_n$ is level respecting, we have that $e_{N+1}\notin R$, while $e_N\in R$. 
\end{proof}

	\section{Proof of the main results} \label{sec:6}
	Before proceeding with the proofs of Theorems \ref{thm:combinatorial}, \ref{thm:logbound}, and \ref{thm:arealaw}, we need one more technical lemma. 
	\begin{lemma}
		\label{lemma:Ax-q1d} Let $R\subset\mathcal{V}_\mathbb{Z}$ be a rectangular configuration with $|R|=N=kM$.

		Then, we have that 
		\beq
		\sum_{\substack{|X| = j \\ X \subset \cV_\mathbb{Z} \setminus \Lambda_\ell}} \exp\left(-\mu \min_{z \in \Z}\left( d_N(\cA_{X,N}, R + z) \right)\right) \leq \frac{8 f(R,\tfrac{\mu}{2})}{1-e^{-\tfrac{\mu}{2}}} .
		\eeq
	\end{lemma}
	
	\begin{proof}
		For convenience, we will write $\Lambda=\Lambda_\ell$.
		We begin by partitioning $\Z$ into three disjoint subsets:
		\begin{enumerate}
			\item $M_1 = \{z \in \Z: (R + z) \subset \cV_\mathbb{Z} \setminus \Lambda  \}$
			\item $M_2 = \{z \in \Z: (R + z) \cap \partial^{in}\Lambda \neq \emptyset  \}$
			\item $M_3 = \{z \in \Z: (R + z) \cap \partial^{in}\Lambda = \emptyset \text{ and }  (R + z) \subset \Lambda \}$
		\end{enumerate}
		Since $R$ is a rectangle, this indeed is a partition of $\Z$. For any fixed value of $j\in\{1,2,\dots,N-1\}$, we have that 
		
		\begin{align}
		&\sum_{\substack{|X| = j \\ X \subset \cV_\mathbb{Z} \setminus \Lambda_\ell}} \exp\left(-\mu \min_{z \in \Z}\left( d_N(\cA_{X,N}, R + z) \right)\right) \\\leq &\sum_{z \in M_1 \sqcup M_2 \sqcup M_3} \sum_{\substack{|X| = j \\ X \subset \cV_\mathbb{Z} \setminus \Lambda_\ell}} \exp\left(-\mu  d_N(\cA_{X,N}, R + z) \right)\:.
		\end{align}
		
		Firstly, let us consider the contributions to the sum coming from $z \in M_1$. We again partition $M_1$ into subsets, $M_1^{r} $ for $r = 1,2, \dots$, where
		\beq
		M_1^r = \{ z \in M_1 : (R + z) \cap L_r(\Lambda) \neq \emptyset, \text{ and } m < r \text{ implies } (R + z) \cap L_m(\Lambda)= \emptyset\} \:.
		\eeq 
		
		So, if $z \in M_1^r$ then $R + z$ and the $r^{th}$ level set of $\Lambda$ have nonempty intersection while the intersection of $R + z$ with any lower level is disjoint. It is clear that $M_1 = \bigsqcup_{r = 1}^{\infty} M_1^r$ and that $|M_1^r| \leq 2$. Notice that if $z \in M_1^r$ then $d(X \cup Z, R+z) \geq |Z|r$. From this observation, we obtain the estimate 
		\begin{align}  \nn& \sum_{z \in M_1}  \sum_{\substack{|X| = j \\ X \subset \cV_\mathbb{Z} \setminus \Lambda_\ell}}  \exp\left(-\mu  d_N(\cA_{X,N}, R + z) \right)\\ = &\sum_{r=1}^\infty \sum_{z \in M_1^r}  \sum_{\substack{|X| = j \\ X \subset \cV_\mathbb{Z} \setminus \Lambda_\ell}}  \exp\left(-\mu  d_N(\cA_{X,N}, R + z) \right)\\
			\leq &\sum_{r=1}^\infty 2e^{-\tfrac{\mu}{2} r(N-j)} \sum_{\substack{|X| = j \\ X \subset\cV_\mathbb{Z} \setminus \Lambda_\ell}} \exp\left(-\tfrac{\mu}{2} d(\cA_{X,N}, R + \hat{z}_r)\right)\label{inequality}\\
			\leq& \sum_{r=1}^\infty 2e^{-\tfrac{\mu}{2} r(N-j)} f(R, \tfrac{\mu}{2}) =  2f(R,\tfrac{\mu}{2})\frac{e^{-\frac{\mu}{2}(N-j)}}{1-e^{-\tfrac{\mu}{2}(N-j)}} \label{eq:M1-est-Q1D}
		\end{align}
In \eqref{inequality}, we estimated $d_N(\cA_{X,N},R+z)\geq \frac{1}{2}d_N(\cA_{X,N},R+z)+\frac{1}{2}(N-j)r$. Moreover, $\hat{z}_r$ denotes the element in $M_1^r$ that maximizes $\exp(-\mu d_N(\mathcal{A}_{X,N},R+z))$ and since $|M_1^r|\leq 2$, Inequality \eqref{inequality} follows.

		Now we focus on the case $z \in M_2$. We immediately get that $|M_2| \leq 2k$. We further partition $M_2$ into mutually disjoint subsets, $M_2^r$, given by 
		\beq
		M_2^r := \{ z \in M_2: |(R+z) \cap (\mathcal{V}_\mathbb{Z} \setminus \Lambda)| = r\}.  
		\eeq 
		
		So if $z\in\M_2^r$, then in $R+z$ consists of $r$ elements which are not in $\Lambda$ and $(N-r)$ elements which are also elements of $\Lambda$. Clearly, we also have that $|M_2^\ell| \leq 2$. From reasoning similar to before we get the estimate
		\begin{align}  \nn \sum_{z \in M_2} \sum_{\substack{ X \subset \mathcal{V}_\mathbb{Z} \setminus \Lambda\\|X| = j }} &\exp\left(-\mu d(\cA_{X,N}, R + z)\right) = \sum_{r=0}^{N-1} \sum_{z \in M_2^r} \sum_{\substack{|X| = j \\ X \subset \cV_\mathbb{Z} \setminus \Lambda}} \exp\left(-\mu d(\cA_{X,N}, R + z)\right)\\ \nn
	&\leq \sum_{\ell=0}^{N-1} 2e^{-\tfrac{\mu}{2} |j-\ell|} f(R,\tfrac{\mu}{2}) \leq 2 f(R,\tfrac{\mu}{2}) \frac{2}{1-e^{-\tfrac{\mu}{2}}}\:.  \label{eq:M2-est-Q1D}
\end{align}

		We can estimate the sum over $z \in M_3$ in a completely analogous way as the one for $z \in M_1$. To this end, one partitions $M_3$ according to the highest (negative) level that $R+z$ intersects non-trivially, which in the end yields the same estimate as \eqref{eq:M1-est-Q1D}.  From combining \eqref{eq:M1-est-Q1D} and \eqref{eq:M2-est-Q1D}, we get
		\begin{align}
			 \sum_{\substack{|X| = j \\ X \subset \cV_\mathbb{Z} \setminus \Lambda}} \exp\left(-\mu \min_{z \in \Z}\left( d(\cA_{X,N}, R + z) \right)\right) 
			&\leq 4 f(R,\tfrac{\mu}{2}) \left(\frac{1}{1-e^{-\tfrac{\mu}{2}}} + \frac{e^{-\frac{\mu}{2}(N-j)}}{1-e^{-\tfrac{\mu}{2}(N-j)}} \right) \nn\\
			&\leq \frac{8 f(R,\tfrac{\mu}{2})}{1-e^{-\tfrac{\mu}{2}}}.
		\end{align}
		
	\end{proof}
	\subsection{Proof of Theorem \ref{thm:combinatorial}}

	\begin{proof}
	To begin with, observe that
		\beq \label{eq:levelgrow}
		L_R(e_n) \geq \frac{1}{2M}(n - N).
		\eeq
		This follows from the fact that for $n\leq N$, we have $L_R(e_n)\leq 0$. For larger values of $n$, note that the size of each level set is given by $L_n(R)=2M$, which implies \eqref{eq:levelgrow}.
	
		Applying Lemma \ref{lem:est-w-levels}, we get the inequality
		\begin{align}
			\sum_{\substack{  X \subset \cV_\mathbb{Z}\\|X| = N}} e^{-\mu d_N(X,R)} &\leq \sum_{j = 0}^{N} \left(\sum_{N < x_1 <\dots < x_j} e^{-\mu\sum_k L_R(e_{x_k}) }\right) \left(\sum_{0 < y_1 < \dots < y_j \leq N} e^{-\mu \sum_k |L_R(e_{y_k})|}\right) \nn \\
			& = \sum_{j=0}^N (I)(II) 
		\end{align}
		We will estimate (I) and (II) by applying \eqref{eq:levelgrow}. First we estimate (I). 
		\begin{align}
			(I)&\leq \sum_{N < x_1 <\dots < x_j} e^{-\frac{\mu}{2M}\sum_{k=1}^j (x_k - N) }\\ \nn 
			&= e^{\frac{\mu j N}{2M}}\sum_{N < x_1 <\dots < x_j} e^{-\frac{\mu}{2M}\sum_{k=1}^j x_k}.
		\end{align}
		
	Estimating the sum by integrals: 
		\begin{align}
			\label{eq:q1dI-integraltest}
			\sum_{N < x_1 <\dots x_j} & e^{-\frac{\mu}{2M}\sum_{k} x_k} \leq \int_N^\infty  dx_1 e^{-\frac{\mu}{2M} x_1} \int_{x_1}^{\infty} dx_2 e^{-\frac{\mu}{2M} x_2} \cdots \int_{x_{j-1}}^\infty dx_j e^{-\frac{\mu}{2M} x_j}\\ \nn
			&= \left(\frac{2M}{\mu}\right)^{j} \int_{\frac{\mu N}{2M}}^\infty  dt_1 e^{-t_1} \int_{t_1}^{\infty} dt_2  e^{-t_2} \cdots \int_{t_{j-1}}^\infty dt_j e^{-t_j} \\
			&= \frac{1}{j!} \left(\frac{2M}{\mu}\right)^{j} e^{-\frac{j \mu N}{2M}}\:.
		\end{align}
		
		Next, we estimate (II). We begin by re-indexing, $z_k = N - y_{j+1-k}$. 
		\begin{align}
			(II)&\leq \sum_{0 < y_1 < \dots < y_j \leq N} e^{-\frac{\mu}{2M} \sum_k (N - y_k)} = \sum_{0 \leq z_1 < \dots < z_j < N} e^{-\frac{\mu}{2M} \sum_k z_k}.
		\end{align}
		We have that 
		\begin{align}
			\label{eq:q1destimateII}
			(II) &\leq \sum_{0 \leq z_1 < \dots < z_j < N} e^{-\tfrac{\mu}{2M}\sum_k z_k}\\ \nn 
			&= \sum_{0 < z_2 < \dots < z_j \leq N} e^{-\tfrac{\mu}{2M}\sum_{k \geq 2} z_k} + \sum_{0 < z_1 < \dots < z_j \leq N} e^{-\tfrac{\mu}{2M}\sum_{k} z_k} = (A) + (B). \\ \nn
		\end{align}
	We estimate both terms, $(A)$ and $(B)$, by integrals:
	\begin{align}
			(A) &\leq  \int_0^\infty  dz_2 e^{-\tfrac{\mu}{2M} z_2} \int_{z_2}^{\infty} dz_3 e^{-\tfrac{\mu}{2M}z_3} \cdots \int_{z_{j-1}}^\infty dz_j e^{-\tfrac{\mu}{2M} z_j} = \frac{1}{(j-1)!}\left(\frac{2M}{\mu}\right)^{j-1}\\ \nn
			(B) &\leq \int_0^\infty  dz_1 e^{-\tfrac{\mu}{2M} z_1} \int_{z_1}^{\infty} dz_2 e^{-\tfrac{\mu}{2M}z_2} \cdots \int_{z_{j-1}}^\infty dz_j e^{-\tfrac{\mu}{2M} z_j} = \frac{1}{j!}\left(\frac{2M}{\mu}\right)^{j}\:.
		\end{align}
		
		Consequently, we get
		\begin{equation}
						(II)\leq (A)+(B) \leq \frac{1}{(j-1)!}\left(\frac{2M}{\mu}\right)^{j-1} + \frac{1}{j!}\left(\frac{2M}{\mu}\right)^{j}= \frac{1}{j!}\left(1 + \frac{j\mu}{2M}\right)\left( \frac{2M}{\mu}\right)^j.
		\end{equation}
		
		Using these estimates, we get that
		
		\begin{align}
			\sum_{j=0}^N (I)(II) &\leq \sum_{j=0}^N \frac{1}{(j!)^2}\left(1 + \frac{j\mu}{2M}\right)\left( \frac{2M}{\mu}\right)^{2j} e^{-\frac{j \mu N}{2M}}\\
			&\leq \sum_{j=0}^N \frac{1}{(j!)^2}\left(1 + \frac{j\mu}{2M}\right)\left( \frac{2M}{\mu}\right)^{2j} \\
			&\leq \exp\left(\frac{4M}{\mu}\right) + \sum_{j=0}^N \frac{j \mu}{2M}\left( \frac{2M}{\mu}\right)^{2j} \frac{1}{(j!)^2}\\
			&\leq \exp\left(\frac{4M}{\mu}\right) + \frac{2M}{\mu} \sum_{j=0}^\infty \left( \frac{2M}{\mu}\right)^{2j} \frac{1}{(j!)^2}
			\leq \left(1 + \frac{2M}{\mu}\right)e^{\frac{4M}{\mu}}\:.
		\end{align}
		
	\end{proof}
	\subsection{Proof of Theorem \ref{thm:logbound}}
	\begin{proof}
		For any $j\in\{1,2,\dots,|\Lambda_\ell^c|-1\}$ we use Proposition \ref{prop:spectralbound} to estimate
		\begin{align} 
			&\sum_{X\subset\Lambda_\ell^c:|X|=j}\|\chi_{\mathcal{A}_{X,N}}\psi\|^{2\alpha}=\sum_{X\subset\Lambda_\ell^c:|X|=j}\|\chi_{\mathcal{A}_{X,N}}Q_\delta^N\psi\|^{2\alpha}\\\leq& \sum_{X\subset\Lambda_\ell^c:|X|=j}\|\chi_{\mathcal{A}_{X,N}}Q_\delta^N\|^{2\alpha}\leq C^{2\alpha}\sum_{X\subset\Lambda_\ell^c:|X|=j}\exp(-2\alpha\mu d_N(\mathcal{A}_{X,N},\cR^N))\:.
			\label{eq:Thm 6.1/1}
		\end{align}
	Since $\Lambda_\ell^c\subset\cV_\mathbb{Z}\setminus\Lambda_\ell$, we get
	\begin{align}
		\eqref{eq:Thm 6.1/1}&\leq C^{2\alpha}\sum_{X\subset\cV_\mathbb{Z}\setminus\Lambda_\ell:|X|=j}\exp(-2\alpha\mu d_N(\mathcal{A}_{X,N},\cR^N))\\&\leq \frac{8C^{2\alpha}f(R,{\mu\alpha})}{1-e^{-{\mu\alpha}}}\leq \frac{8C^{2\alpha}}{1-e^{-\mu\alpha}}\left(1+\frac{2M}{\mu\alpha}\right)e^{\tfrac{4M}{\mu\alpha}}\:,
	\end{align}
where we used Lemma \ref{lemma:Ax-q1d} for the penultimate and Theorem \ref{thm:combinatorial} for the last estimate. Together with \eqref{eq:alphatrace}, this yields the estimate
\begin{align}
	\mathscr{E}(\psi;\ell)\leq\mathscr{E}_\alpha(\psi;\ell)=&\frac{1}{1-\alpha}\log\left(\mbox{tr}(\rho_{\psi;\ell}^\alpha)\right)
	\\\leq& \frac{1}{1-\alpha}\log\left(6+\frac{16 C^{2\alpha}}{1-e^{-\mu\alpha}}\sum_{j=1}^{M\ell-1}\left(1+\frac{2M}{\mu\alpha}\right)e^{\tfrac{4M}{\mu\alpha}}\right)
	\\\leq& \frac{1}{1-\alpha}\log\left(6+\frac{16 C^{2\alpha}}{1-e^{-\mu\alpha}}M\ell\left(1+\frac{2M}{\mu\alpha}\right)e^{\tfrac{4M}{\mu\alpha}}\right)\:,
\end{align}
for any normalized $\psi\in\mathcal{Q}_\delta^N$, for any $N\in\mathfrak{N}$.	Consequently, we get
\begin{align}
	&\limsup_{\ell\rightarrow\infty}\left(\frac{\sup_{N\in\mathfrak{N}}\sup\left\{\mathscr{E}(\psi;\ell):{\psi\in \mathcal{Q}^N_\delta, \|\psi\|=1}\right\}}{\log\ell} \right)\\\leq \frac{1}{1-\alpha}&\limsup_{\ell\rightarrow\infty}\frac{\log\left(6+\frac{16 C^{2\alpha}}{1-e^{-\mu\alpha}}M\ell\left(1+\frac{2M}{\mu\alpha}\right)e^{\tfrac{4M}{\mu\alpha}}\right)}{\log\ell}=\frac{1}{1-\alpha}\:.
\end{align}	
Since this is true for any $\alpha\in(0,1)$ this implies \eqref{eq:logbound}.
	\end{proof}
	\subsection{Proof of Theorem \ref{thm:arealaw}}
	\begin{proof}
		Using Proposition \ref{prop:spectralbound}, we estimate
	\begin{align}
		&\sum_{j=1}^{|\Lambda_\ell^c|-1}\sum_{X\subset\Lambda_\ell^c:|X|=j}\|\chi_{\cA_{X,N}}Q_\delta^N(V_\omega)\|^{2\alpha}\\\leq& \sum_{j=1}^{|\Lambda_\ell^c|-1}\sum_{X\subset\Lambda_\ell^c:|X|=j}\|\chi_{\cA_{X,N}}Q_\delta^N(V_\omega)\|^{\alpha}C^\alpha \exp(-\alpha\mu d_N(\cA_{X,N},\cR^N)) \:.
		\label{eq:randomtrace}
	\end{align}
From Jensen's inequality and Lemma \ref{lemma:largedev} we get:
\begin{align}
	&\mathbb{E}\left[\|\chi_{\cA_{X,N}}Q_\delta^N(V_\omega)\|^{\alpha}C^\alpha \exp(-\alpha\mu d_N(\cA_{X,N},\cR^N))\right]\\ \leq&	C^\alpha \exp\left(-\alpha\mu d_N(\cA_{X,N},\cR^N)\right) \left[\mathbb{E}\left(\|\chi_{\cA_{X,N}}Q_\delta^N(V_\omega)\|\right)\right]^{\alpha}\\ \leq &C^\alpha\tilde{C}^\alpha(\lambda^\alpha)^{|X|}\exp\left(-\alpha\mu d_N(\cA_{X,N},\cR^N)\right)\:.
\end{align}
 Taking expectations of $6+2\times\eqref{eq:randomtrace}$ (cf.\ \eqref{eq:alphatrace}) and using this estimate then yields
\begin{equation}
	\mathbb{E}\left[\sup_{N\in\mathfrak{N}}\:\sup_{\psi}\tr(\rho_{\psi;\ell}^\alpha)\right]\leq 6+2\sum_{j=1}^{|\Lambda_\ell^c|-1}\sum_{X\subset\Lambda_\ell^c:|X|=j}C^\alpha\tilde{C}^\alpha(\lambda^\alpha)^j\exp\left(-\alpha\mu d_N(\cA_{X,N},\cR^N)\right)\:,
\end{equation}
where $``\sup_\psi"$ indicates the supremum over all normalized elements of $\mathcal{Q}_\delta^N(V_\omega)$.
But the sum $\sum_{X}\exp\left(-\alpha\mu d_N(\cA_{X,N},\cR^N)\right)$ can be estimated in the same way as was done in the proof of Theorem \ref{thm:logbound}. We therefore get
\begin{align}
	\mathbb{E}\left[\sup_{N\in\mathfrak{N}}\:\sup_{\psi}\tr(\rho_{\psi;\ell}^\alpha)\right]\leq6+2(C\tilde{C})^\alpha\frac{8}{1-e^{-\tfrac{\mu\alpha}{2}}}\left(1+\frac{4M}{\mu\alpha}\right)e^{\tfrac{8M}{\mu\alpha}}\sum_{j=1}^{M\ell-1}(\lambda^\alpha)^j\\
	\leq 6+2(C\tilde{C})^\alpha\frac{8}{1-e^{-\tfrac{\mu\alpha}{2}}}\left(1+\frac{4M}{\mu\alpha}\right)e^{\tfrac{8M}{\mu\alpha}}\frac{\lambda^\alpha}{1-\lambda^\alpha}=:K_\alpha
\end{align}
where $K_\alpha=K_\alpha(\Delta,\delta,\mathbb{P})$.

Choosing $\alpha=1/2$, we get
\begin{align}
	&\mathbb{E}\left[\sup_{N\in\mathfrak{N}}\sup_\psi \mathscr{E}(\psi;\ell)\right]\leq 	\mathbb{E}\left[\sup_{N\in\mathfrak{N}}\sup_\psi \mathscr{E}_{1/2}(\psi;\ell)\right]=2\mathbb{E}\left[\sup_{N\in\mathfrak{N}}\sup_\psi \log\tr(\rho_{\psi;\ell}^{1/2})\right]\\
	=&2\mathbb{E}\left[\log\left(\sup_{N\in\mathfrak{N}}\sup_\psi \tr(\rho^{1/2}_{\psi;\ell})\right)\right]\leq 2\log\left[\mathbb{E}\left(\sup_{N\in\mathfrak{N}}\sup_\psi \tr(\tr(\rho^{1/2}_{\psi;\ell}))\right)\right]=2\log(K_{1/2})\:.
\end{align}
Since $K_{1/2}$ does not depend on $\ell$, taking the supremum over $\ell$ finishes the proof. 
	\end{proof}
	\begin{acknowledgement}
		CF gratefully acknowledges financial support by the  Simons Foundation through an AMS-Simons Travel Grant.
	\end{acknowledgement}

\end{document}